\documentclass[12pt]{article}

\usepackage{amsmath}
\usepackage{amssymb}
\usepackage{amsfonts}
\usepackage{latexsym}
\usepackage{color}

\catcode `\@=11 \@addtoreset{equation}{section}
\def\theequation{\arabic{section}.\arabic{equation}}
\catcode `\@=12

\newcommand{\be}{\begin{equation}}
\newcommand{\en}{\end{equation}}
\newcommand{\bea}{\begin{eqnarray}}
\newcommand{\ena}{\end{eqnarray}}
\newcommand{\beano}{\begin{eqnarray*}}
\newcommand{\enano}{\end{eqnarray*}}
\newcommand{\bee}{\begin{enumerate}}
\newcommand{\ene}{\end{enumerate}}

\newcommand{\N}{\mathfrak N}

\newcommand{\mc}{\mathcal}

\newcommand{\F}{{\cal F}}

\newcommand{\1}{1 \!\! 1}

\newcommand{\Hil}{\mc H}
\newcommand{\kt}{\rangle}
\newcommand{\br}{\langle}

\newtheorem{thm}{Theorem}

\newenvironment{proof}{\noindent {\bf Proof --}}{\hfill$\square$ \vspace{3mm}\endtrivlist}

\catcode `\@=11 \@addtoreset{equation}{section}
\catcode `\@=12

\begin{document}

\thispagestyle{empty}

\vspace*{2cm}

\begin{center}
{\Large \bf Linear Pseudo-fermions}   \vspace{2cm}\\

{\large F. Bagarello}\\
  Dieetcam,
Facolt\`a di Ingegneria,\\ Universit\`a di Palermo, I-90128  Palermo, Italy\\
e-mail: fabio.bagarello@unipa.it\\
home page: www.unipa.it$\backslash$fabio.bagarello

\end{center}

\vspace*{2cm}

\begin{abstract}
\noindent In a recent series of papers we have analyzed a certain deformation of the canonical commutation relations  producing an interesting
functional structure which has been proved to have some connections with physics, and in particular with quasi-hermitian quantum mechanics.
Here we repeat a similar analysis starting with  the canonical anticommutation relations. We will show that in this case most of the
assumptions needed in the former situation are automatically satisfied, making our construction rather {\em friendly}. We discuss some examples
of our construction, again related to quasi-hermitian quantum mechanics, and the bi-coherent states for the system.

\end{abstract}

\vspace{2cm}


\vfill


\newpage

\section{Introduction}

In a series of papers, \cite{bagpb1}-\cite{bagrev}, we have considered two operators $a$ and $b$, with $b\neq a^\dagger$, acting on a Hilbert
space $\Hil$, and satisfying the  commutation rule $[a,b]=\1$. A nice functional structure has been deduced under suitable assumptions, and
some connections with physics, and in particular with quasi-hermitian quantum mechanics and with the technique of intertwining operators, have
been established. We have called {\em pseudo-bosons} (PB) the particle-like excitations associated to this structure.  A similar analysis has
also been carried out for what we have called {\em nonlinear pseudo-bosons} (NLPB) in \cite{bagnlpb1}-\cite{bagzno2}, and most of the original
results have been recovered also in this more general situation. The analytical treatment of both PB and NLPB  turns out to be particularly
difficult in the case where {\em regularity} is lost, that is, see below, when the biorthogonal bases automatically constructed out of our
strategy are not Riesz bases. In this case, in fact, the intertwining operators appearing in the game (whose square roots are metric operators
in the sense of the literature on quasi-hermitian quantum mechanics, \cite{ben,mosta,zno}) turns out to be unbounded and a large amount of care
should be used to deal properly with their domains, among the other problems.

Here we construct the same structure in a simpler, but still physically very relevant, situation, extending the canonical anticommutation
relation (CAR) in a similar fashion as that producing PB, and considering some mathematical as well as physical aspects of this extension. We
should acknowledge Trifonov and his collaborators for his original idea of dealing with what they call pseudo-fermions (PF), \cite{tripf}. This
idea is considered here adopting a different point of view, going in the direction of our previous works. In particular, we will see that the
mathematical structure here, with respect to that arising from (linear or nonlinear) PB, is more friendly since, being $\Hil$, the Hilbert
space of the theory, intrinsically finite-dimensional, no problem with unbounded operators do appear. This finite-dimensionality of $\Hil$ is
interesting also in view of the many finite-dimensional examples which have been considered, along the years, for quasi-hermitian quantum
mechanics or for similar extensions of {\em ordinary} quantum mechanics\footnote{In the literature many different extensions of ordinary
quantum mechanics do exist which share a common feature,that is the fact that the hamiltonian of the system is not self-adjoint in the {\em
natural} Hilbert space where the model is defined. Just to cite few well-established extensions, Bender is one of the father-founders of the
so-called PT-quantum mechanics, \cite{ben}, Znojil uses the crypto-hermiticity concept, \cite{zno}, Mostafazadeh adoptes what he calls
quasi-hermiticity,\cite{mosta}}. In many examples of these extensions the models considered live in finite dimensional Hilbert spaces,
\cite{ben}-\cite{jon}. This choice has two main consequences: first of all, from a technical point of view one has to do with finite
dimensional matrices. Hence the computations are, at least in principle, simplified. The second, and mathematically more relevant consequence,
is that the observables of the models are all bounded operators, and this highly simplify the rigorous treatment of the system. Just to have an
idea of the differences between bounded and unbounded situations, one could consider the two companion papers \cite{bagzno1} and
\cite{bagzno2}, where this problem has been considered in connection with NLPB and crypto-hermiticity. For this reason, finite dimensional
models turns out to be so important in the present context, and fermionic operators produce a perfect possibility of constructing these kind of
models, as we will see in the rest of the paper.

The paper is organized as follows: in the next section we discuss the linear extension of the CAR, giving rise to what we will call {\em linear
pseudo-fermions}.  Section III is devoted to the examples, while in Section IV we discuss coherent states arising from pseudo-fermions. Section
V contains our conclusions. In the Appendix we briefly consider the existence of an intertwining relation of the kind discussed below for a
rather general 2 by 2 non-self adjoint hamiltonian.

\section{Linear pseudo-fermions}

We begin this section recalling the definition of linear pseudo-bosons and listing some mathematical consequences of this definition. This
preliminary results, which can be found, for instance, in \cite{bagrev}, are useful to keep the paper self-contained. We will show that most of
the  assumptions needed for PB are automatically true for PF and, as a consequence, the new structure we are going to construct is much simpler
than the previous one.

\subsection{A short resume of linear PB}

Let $\Hil$ be a given Hilbert space with scalar product $\left<.,.\right>$ and related norm $\|.\|$. We introduce a pair of operators, $a$ and
$b$,  acting on $\Hil$ and satisfying the  commutation rule \be [a,b]=\1, \label{21} \en where $\1$ is the identity on $\Hil$.  Of course, this
collapses to the canonical commutation rule (CCR)  if $b=a^\dagger$. Calling $D^\infty(X):=\cap_{p\geq0}D(X^p)$  the common domain of all the
powers of the operator $X$, we consider the following:

\vspace{2mm}

{\bf Assumption 1.--} there exists a non-zero $\varphi_{ 0}\in\Hil$ such that $a\varphi_{ 0}=0$, and $\varphi_{ 0}\in D^\infty(b)$.

{\bf Assumption 2.--} there exists a non-zero $\Psi_{ 0}\in\Hil$ such that $b^\dagger\Psi_{ 0}=0$, and $\Psi_{ 0}\in D^\infty(a^\dagger)$.

\vspace{2mm}

Under these assumptions we can introduce the following vectors in $\Hil$:

\be \varphi_{n}=\frac{1}{\sqrt{n!\,}}\,b^{n}\,\varphi_{ 0},\qquad \Psi_{n}=\frac{1}{\sqrt{n!\,}}\,{a^\dagger}^{n}\,\Psi_{ 0}, \label{22}\en
$n=0, 1, 2,\ldots$. Let us now define the unbounded operators $N:=b\,a$ and $\N:=N^\dagger=a^\dagger b^\dagger$. It is possible to check that
$\varphi_{ n}$ belongs to the domain of $N$, $D(N)$, and $\Psi_{ n}\in D(\N)$, for all possible $n$. Moreover, \be N\varphi_{ n}=n\,\varphi_{
n},  \quad \N\Psi_{ n}=n\,\Psi_{ n}. \label{23}\en

Under the above assumptions, and if we chose the normalization of $\Psi_{ 0}$ and $\varphi_{ 0}$ in such a way that $\left<\Psi_{ 0},\varphi_{
0}\right>=1$, we deduce that \be \left<\Psi_{ n},\varphi_{ m}\right>=\delta_{ n,m}. \label{27}\en Then the sets $\F_\Psi=\{\Psi_{ n}\}$ and
$\F_\varphi=\{\varphi_{ n}\}$ are { biorthogonal} and, because of this, the vectors of each set are linearly independent. This suggests to
consider the following

\vspace{2mm}

{\bf Assumption 3.--}  $\F_\Psi$ and $\F_\varphi$ are complete in $\Hil$.

\vspace{2mm}

In particular this means that both $\F_\varphi$ and $\F_\Psi$ are bases of $\Hil$. Let us now introduce the operators $S_\varphi$ and $S_\Psi$
via their action respectively on  $\F_\Psi$ and $\F_\varphi$: \be S_\varphi\Psi_{ n}=\varphi_{ n},\qquad S_\Psi\varphi_{ n}=\Psi_{ n},
\label{213}\en for all $ n$, which also imply that $\Psi_{ n}=(S_\Psi\,S_\varphi)\Psi_{ n}$ and $\varphi_{ n}=(S_\varphi \,S_\Psi)\varphi_{
n}$, for all $ n$. Hence \be S_\Psi\,S_\varphi=S_\varphi\,S_\Psi=\1 \quad \Rightarrow \quad S_\Psi=S_\varphi^{-1}, \label{214}\en at least if
these operators are bounded. If they are not bounded, on the other hand, this is not guaranteed, \cite{bagzno2}. In other words, both $S_\Psi$
and $S_\varphi$ are invertible and one is the inverse of the other. Furthermore, they are both positive, well defined and symmetric,
\cite{bagpb1}. Moreover, it is possible to write these operators using the bra-ket notation as \be S_\varphi=\sum_{ n}\, |\varphi_{
n}><\varphi_{ n}|,\qquad S_\Psi=\sum_{ n} \,|\Psi_{ n}><\Psi_{ n}|. \label{212}\en Whenever $S_\varphi$ and $S_\Psi$ are bounded these series
are uniformly convergent and, as a consequence, $\F_\varphi$ and $\F_\Psi$ turn out to be Riesz bases of $\Hil$.

In the literature we have called {\em regular} those PB for which $\F_\varphi$ and $\F_\Psi$ are Riesz bases. This is not always true, as shown
by the physical examples discussed, for instance, in \cite{bagpb4}. For this reason in our previous literature we have also added the following

\vspace{2mm}

{\bf Assumption 4.--}  $\F_\Psi$ and $\F_\varphi$ are Riesz bases for $\Hil$.

\vspace{2mm}

We have proved that this assumption is satisfied if, and only if, the operators $S_\varphi$ and $S_\Psi$ are both bounded, property which, for
concrete models, is quite often violated: in physically motivated models, PB which are not regular seem to be the {\em natural ones}, and we
have necessarily to deal with unbounded operators.

It is easy to check that \be S_\Psi\,N=\N\,S_\Psi \quad \mbox{ and }\quad N\,S_\varphi=S_\varphi\,\N. \label{219}\en  This is in agreement with
the fact that the eigenvalues of $N$ and $\N$ coincide and that their eigenvectors are related by the operators $S_\varphi$ and $S_\Psi$, as
suggested by the literature on intertwining operators.

\subsection{Linear pseudo-fermions}

The problem of domains of possible unbounded operators does not exist for fermions, making all the story much simpler. This section is devoted
to show how CAR can be modified as we did for CCR, and how this extension looks more {\em natural and safe} than that for PB. As for linear PB
the starting point is a modification of the CAR $\{c,c^\dagger\}=c\,c^\dagger+c^\dagger\,c=\1$, $\{c,c\}=\{c^\dagger,c^\dagger\}=0$, between
two operators, $c$ and $c^\dagger$, acting on a two-dimensional Hilbert space $\Hil$. The CAR are replaced here by the following rules: \be
\{a,b\}=\1, \quad \{a,a\}=0,\quad \{b,b\}=0, \label{220}\en where the relevant situation is when $b\neq a^\dagger$. Compared with Assumptions
1-4 for PB, the only assumptions we need to require now are the following

\begin{itemize}

\item {\bf p1.} a non zero vector $\varphi_0$ exists in $\Hil$ such that $a\,\varphi_0=0$.

\item {\bf p2.} a non zero vector $\Psi_0$ exists in $\Hil$ such that $b^\dagger\,\Psi_0=0$.

\end{itemize}

Under these two natural conditions it is possible to recover similar results as those for PB. In particular, we first introduce the following
non zero vectors \be \varphi_1:=b\varphi_0,\quad \Psi_1=a^\dagger \Psi_0, \label{221}\en as well as the non self-adjoint operators \be
N=ba,\quad \N=N^\dagger=a^\dagger b^\dagger. \label{222}\en We further introduce the self-adjoint operators $S_\varphi$ and $S_\Psi$ via their
action on a generic $f\in\Hil$: \be S_\varphi f=\sum_{n=0}^1\br\varphi_n,f\kt\,\varphi_n, \quad S_\Psi f=\sum_{n=0}^1\br\Psi_n,f\kt\,\Psi_n.
\label{223}\en Hence we get the following results, similar in part to those for PB,  whose proofs are straightforward and will not be given
here:

\begin{enumerate}

\item \be a\varphi_1=\varphi_0,\quad b^\dagger\Psi_1=\Psi_0.\label{224}\en

\item \be N\varphi_n=n\varphi_n,\quad \N\Psi_n=n\Psi_n,\label{225}\en
for $n=0,1$.

\item If the normalizations of $\varphi_0$ and $\Psi_0$ are chosen in such a way that $\left<\varphi_0,\Psi_0\right>=1$,
then \be \left<\varphi_k,\Psi_n\right>=\delta_{k,n},\label{226}\en for $k,n=0,1$.

\item $S_\varphi$ and $S_\Psi$ are bounded, strictly positive, self-adjoint, and invertible. They satisfy
\be \|S_\varphi\|\leq\|\varphi_0\|^2+\|\varphi_1\|^2, \quad \|S_\Psi\|\leq\|\Psi_0\|^2+\|\Psi_1\|^2,\label{227}\en \be S_\varphi
\Psi_n=\varphi_n,\qquad S_\Psi \varphi_n=\Psi_n,\label{228}\en for $n=0,1$, as well as $S_\varphi=S_\Psi^{-1}$ and the following
intertwining relations \be S_\Psi N=\N S_\Psi,\qquad S_\varphi \N=N S_\varphi.\label{229}\en

\end{enumerate}

The above formulas show that (i) $N$ and $\N$ behave as fermionic number operators, having eigenvalues 0 and 1; (ii) their related eigenvectors
are respectively the vectors in $\F_\varphi=\{\varphi_0,\varphi_1\}$ and $\F_\Psi=\{\Psi_0,\Psi_1\}$; (iii) $a$ and $b^\dagger$ are lowering
operators for $\F_\varphi$ and $\F_\Psi$ respectively; (iv) $b$ and $a^\dagger$ are rising operators for $\F_\varphi$ and $\F_\Psi$
respectively; (v) the two sets $\F_\varphi$ and $\F_\Psi$ are biorthonormal; (vi) the {\em very well-behaved} operators $S_\varphi$ and
$S_\Psi$ maps $\F_\varphi$ in $\F_\Psi$ and viceversa; (vii) $S_\varphi$ and $S_\Psi$ intertwine between operators which are not self-adjoint,
in the very same way as they do for PB.

It is clear that we don't need to add any condition on the possibility of computing, for instance, $b\,\varphi_0$, as we had to do for PB, see
Assumption 1. In fact,  we can always act on a two-dimensional vector with a two-by-two matrix! Also, we don't need to check (or to ask for)
Assumption 3, since this is automatically satisfied: being biorthogonal, the vectors of both $\F_\varphi$ and $\F_\Psi$ are linearly
independent. Hence $\varphi_0$ and $\varphi_1$ are linearly independent in a two-dimensional Hilbert space, so that $\F_\varphi$ is a basis for
$\Hil$. The same conclusion obviously applies to $\F_\Psi$. We will show in a moment that both these sets are also Riesz bases, so that
Assumption 4 for PB is automatically satisfied as well. This will appear to be a consequence of the properties listed above for $S_\varphi$ and
$S_\Psi$.

In \cite{bagpbJPA} we have analyzed the relations existing between PB, regular or not, with ordinary bosons, and we have shown how slightly
more complicated methods of functional analysis are needed when regularity is lost, due to the fact that the operators $S_\varphi$ and $S_\Psi$
turn out to be unbounded. Therefore, it is not surprising that this kind of difficulties do not appear here since, as the inequalities in
(\ref{227}) show, $S_\varphi$ and $S_\Psi$ are bounded operators. More explicitly, the two main theorems proved in \cite{bagpbJPA} assume now
the following simpler form:

\begin{thm} Let $c$ and $T=T^\dagger$ be two operators on $\Hil$ such that $\{c,c^\dagger\}=\1$, $c^2=0$, and $T>0$. Then, defining
\be a=T\,c\,T^{-1},\quad b=T\,c^\dagger\,T^{-1},\label{230}\en these operators satisfy (\ref{220}) as well as properties {\bf p1} and {\bf p2}.

Viceversa, given two operators $a$ and $b$ acting on $\Hil$, satisfying (\ref{220}) and properties {\bf p1} and {\bf p2},  it is possible to
define two operators, $c$ and $T$, such that $\{c,c^\dagger\}=\1$, $c^2=0$, $T=T^\dagger$ is strictly positive, and (\ref{230}) holds.

\end{thm}

\begin{proof}

The proof of the first part of the theorem is trivial and will not be given here. As for the second part, we start with the following remark:
since $S_\Psi$ is positive and invertible, the operators $S_\Psi^{\pm1/2}$ are both well defined. Hence we can define, using $\F_\varphi$ and
$S_\Psi^{1/2}$, another family $\F_f=\{f_0,f_1\}$ of vectors $f_n:=S_\Psi^{1/2}\varphi_n$, $n=0,1$. $\F_f$ is an orthonormal basis for $\Hil$,
so that we can naturally define an operator $c$ via its action on $f_0$ and $f_1$: we put $c\,f_0=0$ and $c\,f_1=f_0$, so that $c^\dagger
f_0=f_1$ and $c^\dagger f_1=0$. It is easy to check that $\{c,c^\dagger\}=\1$ and that $c^2=0$. Moreover, since
$S_\Psi^{-1/2}c\,S_\Psi^{1/2}\varphi_0=0$ and $S_\Psi^{-1/2}c\,S_\Psi^{1/2}\varphi_1=\varphi_0$,  we deduce that
$a=S_\Psi^{-1/2}c\,S_\Psi^{1/2}$. Analogously we find that $b=S_\Psi^{-1/2}c^\dagger\,S_\Psi^{1/2}$, so that we can now identify $T$ in
(\ref{230}) with $S_\Psi^{-1/2}$, which is strictly positive.

\end{proof}

A first consequence of this theorem is that, since $\F_\varphi$ is the image of the orthonormal basis $\F_f$ via a bounded operator, with
bounded inverse, $S_\Psi^{-1/2}$, $\F_\varphi$ is a Riesz basis. A second consequence is that, introducing the self-adjoint number operator for
the fermionic operators, $N_0:=c^\dagger c$, this can be related to both $N$ and $\N$: \be N=S_\Psi^{-1/2}N_0\,S_\Psi^{1/2}, \quad
\N=S_\Psi^{1/2}N_0\,S_\Psi^{-1/2}, \label{231}\en which can be written as the following intertwining relations: $S_\Psi^{1/2}N_0=\N
S_\Psi^{1/2}$, $S_\Psi^{1/2}N=N_0 S_\Psi^{1/2}$. Putting together these equations we can also recover (\ref{229}).

\section{Examples}

In this section we will discuss some examples of our general framework, starting with the easiest one and discussing finally a rather general
application.

\subsection{A one parameter extension of CAR}

This first example is motivated by what we have done in \cite{bagpb1} for PB, where ordinary bosonic operators $A$ and $A^\dagger$,
$[A,A^\dagger]=\1$, were used to construct two linear combinations, $X=\alpha_x A+\beta_x A^\dagger$ and $Y=\alpha_y A+\beta_y A^\dagger$ with
$X\neq Y^\dagger$, satisfying the commutation rule $[X,Y]=\1$. We repeat here a similar construction, and we show that, in the present
situation, we cannot go too far.

Let $c$ and $c^\dagger$ be two fermionic operators: $\{c,c^\dagger\}=\1$, $c^2=0$, and let us introduce two new operators $a=\alpha_a c+\beta_a
c^\dagger$ and $b=\alpha_b c+\beta_b c^\dagger$, with coefficients chosen real and such that $a^\dagger\neq b$. The pseudo-fermionic rules
$\{a,b\}=\1$, $a^2=0$ and $b^2=0$ impose the following conditions: $\alpha_a\beta_b+\beta_a\alpha_b=1$, $\alpha_a\beta_a=0$ and
$\alpha_b\beta_b=0$. These two last equalities, which didn't appear when dealing with PB, makes this example almost trivial. Indeed, a possible
choice for the operators $a$ and $b$ is the following: $a=\beta_a c^\dagger$, $b=\frac{1}{\beta_a}\,c$, which implies that $N=\1-N_0$. Also,
calling as usual $f_0$ and $f_1$ the eigenstates of $N_0=c^\dagger c$ with eigenvalues 0 and 1, we find $\varphi_0=k\,f_1$,
$\varphi_1=b\varphi_0=\frac{k}{\beta_a}\,f_0$, and $\Psi_0=\frac{1}{\overline{k}}\,f_1$,
$\Psi_1=a^\dagger\Psi_0=\frac{\beta_a}{\overline{k}}\,f_0$ for all possible choice of complex $k\neq0$. We see that, as expected, $\F_\varphi$
and $\F_\Psi$ are biorthonormal bases of $\Hil$ (but they are not very different from $\F_f$). Needless to say, another possible choices of
$\alpha$'s and $\beta$'s is possible, producing $a=\alpha_a c$, $b=\frac{1}{\alpha_a}\,c^\dagger$.

\subsection{An example from the theorem}

Let $k, \alpha$ be two real numbers, with $k>0$ and $\alpha\in]-1,1[$. Let us introduce the following matrices
$$
a=\frac{1}{1-\alpha^2}\left(
    \begin{array}{cc}
      -\alpha & 1 \\
      -\alpha^2 & \alpha \\
    \end{array}
  \right),\qquad
  b=\frac{1}{1-\alpha^2}\left(
    \begin{array}{cc}
      \alpha & -\alpha^2 \\
      1 & -\alpha \\
    \end{array}
  \right),
$$
and the following vectors
$$
\varphi_0=k\left(
             \begin{array}{c}
               1 \\
               \alpha \\
             \end{array}
           \right),\quad
           \varphi_1=k\left(
             \begin{array}{c}
               \alpha \\
               1 \\
             \end{array}
           \right),\quad
\Psi_0=\frac{1}{k(1-\alpha^2)}\left(
             \begin{array}{c}
               1 \\
               -\alpha \\
             \end{array}
           \right),\quad
\Psi_1=\frac{1}{k(1-\alpha^2)}\left(
             \begin{array}{c}
               -\alpha \\
               1 \\
             \end{array}
           \right).
$$
It is easy to see that $\F_\varphi$ and $\F_\Psi$ are biorthonormal bases of $\Hil$, and that $a$ and $b$ are pseudo-fermionic operators in the
sense that $\{a,b\}=\1$, $a^2=0$ and $b^2=0$. Moreover $a\varphi_0=b^\dagger\Psi_0=0$,  $b\varphi_0=\varphi_1$ and $a^\dagger\Psi_0=\Psi_1$.
The operators $S_\varphi$ and $S_\Psi$ can now be easily deduced by computing, for example, $S_\varphi
f=\left<\varphi_0,f\right>\varphi_0+\left<\varphi_1,f\right>\varphi_1$, for a generic vector $f\in\Hil$:
$$
S_\varphi=k^2\left(
               \begin{array}{cc}
                 1+\alpha^2 & 2\alpha \\
                 2\alpha & 1+\alpha^2 \\
               \end{array}
             \right), \quad
             S_\Psi=\frac{1}{k^2(1-\alpha^2)^2}\left(
               \begin{array}{cc}
                 1+\alpha^2 & -2\alpha \\
                 -2\alpha & 1+\alpha^2 \\
               \end{array}
             \right).
$$
These are clearly self-adjoint matrices and one is the inverse of the other: $S_\varphi^{-1}=S_\Psi$. The square root of $S_\varphi$ turn out
to be the following matrix:
$$
S_\varphi^{1/2}=k\left(
                   \begin{array}{cc}
                     1 & \alpha \\
                     \alpha & 1 \\
                   \end{array}
                 \right), \quad\Rightarrow \quad S_\varphi^{-1/2}= S_\Psi^{1/2}=\frac{1}{k(1-\alpha^2)}\left(
                   \begin{array}{cc}
                     1 & -\alpha \\
                     -\alpha & 1 \\
                   \end{array}
                 \right).
$$
Following the proof of the theorem above, and formula (\ref{230}) in particular with $T$ identified with $S_\varphi^{1/2}$, we deduce that
$$
S_\varphi^{-1/2}a\,S_\varphi^{1/2}=\left(
                                     \begin{array}{cc}
                                       0 & 1 \\
                                       0 & 0 \\
                                     \end{array}
                                   \right)=:c,\quad
                                   S_\varphi^{-1/2}b\,S_\varphi^{1/2}=\left(
                                     \begin{array}{cc}
                                       0 & 0 \\
                                       1 & 0 \\
                                     \end{array}
                                   \right)=:c^\dagger.
$$
Hence we recover the ordinary fermionic operators, as claimed by the theorem. Moreover $S_\varphi^{-1/2}\varphi_n=f_n$, $n=0,1$, where
$$f_0=\left(
        \begin{array}{c}
          1 \\
          0 \\
        \end{array}
      \right),\quad
f_1=\left(
        \begin{array}{c}
          0 \\
          1 \\
        \end{array}
      \right).
$$
Therefore, as expected, PF turn out to be similar to ordinary fermions.

\subsection{An example from the literature}

In 2007, in \cite{tripf}, an effective non self-adjoint hamiltonian describing a two level atom interacting with an electromagnetic field was
analyzed in connection with pseudo-hermitian systems. This example is meant to show that the cited model can be very naturally rewritten in
terms of pseudo-fermionic operators, and that the structure previously described naturally arises. The starting point is the Schr\"odinger
equation \be i\dot\Phi(t)=H_{eff}\Phi(t), \qquad H_{eff}=\frac{1}{2}\left(
                                                       \begin{array}{cc}
                                                         -i\delta & \overline{\omega} \\
                                                         \omega & i\delta \\
                                                       \end{array}
                                                     \right).
\label{triex}\en Here $\delta$ is a real quantity, related to the decay rates for the two levels, while the complex parameter $\omega$
characterizes the radiation-atom interaction. We refer to \cite{tripf} for further details. It is clear that $H_{eff}\neq H_{eff}^\dagger$. It
is convenient to write $\omega=|\omega|e^{i\theta}$. Then, we introduce the  operators
$$
a=\frac{1}{2\Omega}\left(
                     \begin{array}{cc}
                       -|\omega| & -e^{-i\theta}(\Omega+i\delta) \\
                       e^{i\theta}(\Omega-i\delta) & |\omega| \\
                     \end{array}
                   \right), \quad
b=\frac{1}{2\Omega}\left(
                     \begin{array}{cc}
                       -|\omega| & e^{-i\theta}(\Omega-i\delta) \\
                       -e^{i\theta}(\Omega+i\delta) & |\omega| \\
                     \end{array}
                   \right).
$$
Here $\Omega=\sqrt{|\omega|^2-\delta^2}$, which we will assume here to be real and strictly positive. A direct computation shows that
$\{a,b\}=\1$, $a^2=b^2=0$. Hence $a$ and $b$ are pseudo-fermionic operators. Moreover, $H_{eff}$ can be written in terms of these operators as
$H_{eff}=\Omega\left(ba-\frac{1}{2}\1\right)$.

To recover the pseudo-fermionic structure we first need to check whether a non-zero vector $\varphi_0$ annihilated by $a$ does exist. It is
easy to find such a vector, as well as a second vector $\Psi_0$ annihilated by $b^\dagger$. These two vectors are
$$
\varphi_0=k\left(
             \begin{array}{c}
               1 \\
               -\,\frac{e^{i\theta}(\Omega-i\delta)}{|\omega|} \\
             \end{array}
           \right),\qquad
\Psi_0=k'\left(
             \begin{array}{c}
               1 \\
               -\,\frac{e^{i\theta}(\Omega+i\delta)}{|\omega|} \\
             \end{array}
           \right),
$$
where $k$ and $k'$ are normalization constants, partially fixed by the requirement that
$\left<\varphi_0,\Psi_0\right>=\overline{k}\,k'\left(1+\frac{1}{|\omega|^2}(\Omega+i\delta)^2\right)=1$. Following what we have done in Section
II we also introduce the following vectors
$$
\varphi_1=b\varphi_0=k\left(
             \begin{array}{c}
               \frac{i\delta-\Omega}{|\omega|} \\
               -e^{i\theta} \\
             \end{array}
           \right),\qquad
\Psi_1=a^\dagger\Psi_0=k'\left(
             \begin{array}{c}
               \frac{-i\delta-\Omega}{|\omega|} \\
               -e^{i\theta} \\
             \end{array}
           \right).
$$
It is now easy to check that $\F_\varphi$ and $\F_\Psi$ are biorthonormal bases of $\Hil$, and we can also check that
$$
H_{eff}\varphi_0=-\,\frac{\Omega}{2}\,\varphi_0,\quad H_{eff}\varphi_1=\frac{\Omega}{2}\,\varphi_1, \quad
 H_{eff}^\dagger\Psi_0=-\,\frac{\Omega}{2}\,\Psi_0,\quad H_{eff}^\dagger\Psi_1=\frac{\Omega}{2}\,\Psi_1.
$$
Therefore $H_{eff}$ and $H_{eff}^\dagger$ are isospectrals, as expected. To carry on our analysis we now compute $S_\varphi$ and $S_\Psi$,
which are found to be
$$
S_\varphi=2|k|^2\left(
                  \begin{array}{cc}
                    1 & \frac{-i\delta}{|\omega|}\,e^{-i\theta} \\
                    \frac{i\delta}{|\omega|}\,e^{i\theta} & 1 \\
                  \end{array}
                \right),\quad
S_\Psi=\frac{|\omega|^2}{2|k|^2\Omega^2}\left(
                  \begin{array}{cc}
                    1 & \frac{i\delta}{|\omega|}\,e^{-i\theta} \\
                    \frac{-i\delta}{|\omega|}\,e^{i\theta} & 1 \\
                  \end{array}
                \right),
$$
and turn out to be one the inverse of the other. They are also, under our assumption $\Omega>0$, positive definite matrices, as they should.
Using now the theorem proved in Section II, we can use $S_\varphi^{\pm1/2}$ to define two {\em standard} fermion operators  $c$ and
$c^\dagger$, and their related number operator $N_0=c^\dagger c$, out of $a$ and $b$. Hence we easily find that
$$H_{eff}=S_\varphi^{1/2}\,h\,S_\varphi^{-1/2},$$
where $h=\Omega\left(c^\dagger c-\frac{1}{2}\1\right)$ is a self adjoint operator. This shows that the effective hamiltonian $H_{eff}$ is
similar to a self-adjoint operator, suggesting that, at least for this model, the appearance of a non self-adjoint hamiltonian is related to
the choice of a {\em natural but wrong} scalar product in the Hilbert space $\Hil$. In fact, replacing the standard scalar product
$\left<f,g\right>$ with a new one, $\left<f,g\right>_S=\left<S_\varphi^{-1/2}f,S_\varphi^{-1/2}g\right>$, would make $H_{eff}$ self adjoint:
for all $f, g\in \Hil$ indeed we find
$$
\left<H_{eff}\,f,g\right>_S=\left<S_\varphi^{-1/2}\left(S_\varphi^{1/2}h\,S_\varphi^{-1/2}\right)f,S_\varphi^{-1/2}g\right>=
\left<S_\varphi^{-1/2}f,h\,S_\varphi^{-1/2}g\right>=$$
$$=\left<S_\varphi^{-1/2}f,S_\varphi^{-1/2}\left(S_\varphi^{1/2}h\,S_\varphi^{-1/2}\right)g\right>=\left<f,H_{eff}\,g\right>_S.
$$

\subsection{A general construction}

The starting point for this example is a pair of biorthonormal bases, $\F_\varphi=\{\varphi_0,\varphi_1\}$ and $\F_\Psi=\{\Psi_0,\Psi_1\}$,
$\left<\varphi_j,\Psi_k\right>=\delta_{j,k}$, which we use to construct two non self-adjoint operators
$$
H=\epsilon_0|\varphi_0\left>\right<\Psi_0|+\epsilon_1|\varphi_1\left>\right<\Psi_1|, \quad
H^\dag=\epsilon_0|\Psi_0\left>\right<\varphi_0|+\epsilon_1|\Psi_1\left>\right<\varphi_1|,
$$
where $\left(|f\left>\right<g|\right)h=\left<g,h\right>f$, for all $f,g,h\in\Hil$. Here we take $\epsilon_0$ and $\epsilon_1$ real. It is clear
that $H\varphi_n=\epsilon_n\varphi_n$ and $H^\dagger\Psi_n=\epsilon_n\Psi_n$, for $n=0,1$. It is also clear that, for instance
$|\varphi_0\left>\right<\Psi_0|+|\varphi_1\left>\right<\Psi_1|=\1$, while $a=|\varphi_0\left>\right<\Psi_1|$ and
$b=|\varphi_1\left>\right<\Psi_0|$. Hence we can write
$$
H=\left(\epsilon_1-\epsilon_0\right)b\,a+\epsilon_0\1,
$$
which is interesting for us since it implies that $H-\epsilon_0\1$ can be factorized. Now, for concreteness' sake, let us consider the
following expressions for $\varphi_k$ and $\Psi_k$:
$$
\varphi_0=\left(
            \begin{array}{c}
              \cosh(\theta) \\
              \sinh(\theta)e^{-i\varphi} \\
            \end{array}
          \right), \qquad
\varphi_1=\left(
            \begin{array}{c}
              \sinh(\theta)e^{i\varphi} \\
              \cosh(\theta) \\
            \end{array}
          \right),
$$
$$
\Psi_0=\left(
            \begin{array}{c}
              \cosh(\theta) \\
              -\sinh(\theta)e^{-i\varphi} \\
            \end{array}
          \right), \qquad
\Psi_1=\left(
            \begin{array}{c}
              -\sinh(\theta)e^{i\varphi} \\
              \cosh(\theta) \\
            \end{array}
          \right).
$$
With this choice, the operator $H$ looks like
$$
H=\left(
    \begin{array}{cc}
      \epsilon_0\cosh^2(\theta)-\epsilon_1\sinh^2(\theta) & (\epsilon_1-\epsilon_0)\sinh(\theta)\cosh(\theta)e^{i\varphi} \\
      -(\epsilon_1-\epsilon_0)\sinh(\theta)\cosh(\theta)e^{-i\varphi} & -\epsilon_0\sinh^2(\theta)+\epsilon_1\cosh^2(\theta) \\
    \end{array}
  \right),
$$
which is manifestly non self-adjoint. The matrix expressions for $a$ and $b$ turn out to be
$$
a=\left(
    \begin{array}{cc}
      -\sinh(\theta)\cosh(\theta)e^{-i\varphi} & \cosh^2(\theta) \\
      -\sinh^2(\theta)e^{-2i\varphi} & \sinh(\theta)\cosh(\theta)e^{-i\varphi} \\
    \end{array}
  \right),
$$
$$
b=\left(
    \begin{array}{cc}
      \sinh(\theta)\cosh(\theta)e^{i\varphi} & -\sinh^2(\theta)e^{2i\varphi} \\
       \cosh^2(\theta) & -\sinh(\theta)\cosh(\theta)e^{i\varphi} \\
    \end{array}
  \right),
$$
which explicitly satisfy $\{a,b\}=\1$, $a^2=b^2=0$. Moreover we get
$$
S_\Psi=\left(
    \begin{array}{cc}
      \cosh(2\theta) & -\sinh(2\theta)e^{i\varphi} \\
       -\sinh(2\theta)e^{-i\varphi} & \cosh(2\theta)\\
    \end{array}
  \right),
$$
which satisfies the intertwining relation $S_\Psi H=H^\dagger S_\Psi$. The self-adjoint counterpart of $H$ looks now very simple:
$$
h=S_\Psi^{1/2}H\,S_\Psi^{-1/2}=\left(
                                 \begin{array}{cc}
                                   \epsilon_0 & 0 \\
                                   0 & \epsilon_1 \\
                                 \end{array}
                               \right),
$$
which again suggests that $H$, non self-adjoint, is nothing but the operator $h$ in a Hilbert space with a different scalar product.

\section{Bi-coherent states}

The existence of bi-coherent states for pseudo-fermions has been considered in \cite{tripf}, but from a different perspective with respect to
the one we are interested here. We will consider this problem now putting in evidence the pseudo-fermionic structure.

Coherent states (CS) for fermion or for boson operators can be introduced in several inequivalent ways, depending on which aspect of CS we are
interested in. In particular, we will adopt the probably most common definition of what a CS is, i.e. a state $\Phi_\xi$ which is eigenvector
of the annihilation operator $c$, with $\{c,c^\dagger\}=\1$ and $c^2=0$, with eigenvalue $\xi$: $c\Phi_\xi=\xi\Phi_\xi$. Is is well known that
$\xi$ must be a grassmann number satisfying
$$
\{\xi^\sharp,\xi^\sharp\}=\{\xi^\sharp,c^\sharp\}=0,
$$
where $\xi^\sharp$ and $c^\sharp$ stand respectively for $\xi$ or $\overline{\xi}$ and for $c$ or $c^\dagger$. Calling, as usual, $f_0$ and
$f_1$ the eigenstates of $N_0=c^\dagger c$ with eigenvalues 0 and 1, equivalent forms of $\Phi_\xi$ are the following:
$$
\Phi_\xi=e^{c^\dagger \xi-\overline{\xi}c}f_0=e^{-\frac{\overline{\xi}\xi}{2}}\left(f_0+f_1\xi\right).
$$
These states {\em solve the unity} in the following sense: using the grasmmmann integration rules $\int d\xi=\int d{\overline{\xi}}=0$ and
$\int \xi d\xi=\int \overline{\xi} d{\overline{\xi}}=1$, we have $$\int |\Phi_\xi\left>\right<\Phi_\xi|\,d\xi
d{\overline{\xi}}=|f_0\left>\right<f_0|+|f_1\left>\right<f_1|=\1.$$

Let us now define two new vectors: $\varphi_\xi=S_\Psi^{-1/2}\Phi_\xi$ and $\Psi_\xi=S_\Psi^{1/2}\Phi_\xi$. It is easy to check that
$$
a\varphi_\xi=\xi\varphi_\xi,\quad b^\dagger\Psi_\xi=\xi\Psi_\xi,
$$
so that they are eigenstates of the two lowering pseudo-fermionic operators. These vectors can be written in different ways. In particular they
can be written in terms of $\F_\varphi$ and $\F_\Psi$ as follows:
$$
\varphi_\xi=S_\Psi^{-1/2}\Phi_\xi=S_\Psi^{-1/2}e^{c^\dagger \xi-\overline{\xi}c}f_0=\left(1-\frac{1}{2}\overline{\xi}\xi\right)
\varphi_0-\xi\varphi_1,
$$
and
$$
\Psi_\xi=S_\Psi^{1/2}\Phi_\xi=S_\Psi^{1/2}e^{c^\dagger \xi-\overline{\xi}c}f_0=\left(1-\frac{1}{2}\overline{\xi}\xi\right)
\Psi_0-\xi\Psi_1.
$$
Using these simple expressions for $\varphi_\xi$ and $\Psi_\xi$ it is now an easy exercise to check explicitly that
$a\varphi_\xi=\xi\varphi_\xi$ and $b^\dagger\Psi_\xi=\xi\Psi_\xi$. Moreover we find, using the above rules for the grassmann integration,
$$
\int |\varphi_\xi\left>\right<\Psi_\xi|\,d\xi d{\overline{\xi}}=|\varphi_0\left>\right<\Psi_0|+|\varphi_1\left>\right<\Psi_1|=\1,
$$
so that, used together, $\varphi_\xi$ and $\Psi_\xi$ produce a resolution of the identity. This is the reason why these states are called {\em
bi-coherent}.

\section{Conclusions}

We have shown how the CAR can be modified to get two families of biorthonormal vectors spanning all of $\Hil$ and producing some interesting
intertwining relations. The results of our construction are similar to those obtained for PB, even if the assumptions needed can be
significantly relaxed here. We have considered the relations between this structure and some non self-adjoint hamiltonians, showing that these
became self-adjoint with a proper definition of the scalar product which involves the operators $S_\varphi$ and $S_\Psi$. Finally, we have also
discussed how bi-coherent states can be introduced.

\section*{Acknowledgements}

I would like to thank very, very much Prof. Francesco Oliveri, which is always ready to support me with any problem in numerical computations
(and related stuff). I also like to thank Prof. Kibler for giving me the idea during a conference in Prague, maybe just for joke, to extend my
PB to PF. Dear Maurice, this paper is exactly this extension!  The author also acknowledges financial support by the MIUR.

\vspace{8mm}

 \appendix

\renewcommand{\theequation}{\Alph{section}.\arabic{equation}}

 \section{\hspace{-.7cm}ppendix:  When does (\ref{229}) hold true?}

In this appendix we consider the following problem: given the following non self-adjoint {\em hamiltonian}
 $$
H=\left(
    \begin{array}{cc}
      a & b \\
      c & -a \\
    \end{array}
  \right),
 $$
with $a, b, c\in \Bbb{C}$, when does a positive-defined, self-adjoint matrix $S$  exist such that $S\,H=H^\dagger\,S$?

We first observe that $H$ is traceless. This is not a big constraint, since with a simple (non-unitary) transformation we can get such an
hamiltonian starting with a non-traceless one. Secondly, because of our requirements on $S$, this must be of the following form:
$$
S=\left(
    \begin{array}{cc}
      \sigma_{11} & s \\
      \overline{s} & \sigma_{22} \\
    \end{array}
  \right),
$$
where $s\in\Bbb{C}$, while $\sigma_{11}, \sigma_{22}\in\Bbb{R}$. Moreover, since $S$ has to be positive, the following must be satisfied:
$\sigma_{11}>0$ and $\sigma_{11}\sigma_{22}-|s|^2>0$, which automatically imply that also $\sigma_{22}$ must be positive.

Equation $S\,H=H^\dagger\,S$ can be rewritten as the following matrix equation
$$
X\Phi=0,\,\mbox{ where }\, X=\left(
                  \begin{array}{cccc}
                    \Im(a) & 0 &  \Im(c) & \Re(c) \\
                    0 & \Im(a) & -\Im(b) & \Re(b) \\
                    \Re(b) & -\Re(c) & -2\Re(a) & 0 \\
                    \Im(b) & \Im(c) & 0 & -2\Re(a) \\
                  \end{array}
                \right), \,\mbox{ and }\, \Phi=\left(
                                                 \begin{array}{c}
                                                    \sigma_{11} \\
                                                    \sigma_{22} \\
                                                   \Re(s) \\
                                                   \Im(s) \\
                                                 \end{array}
                                               \right).
$$
It is now clear that a non trivial solution can exist only if $\det(X)=0$. This produces the following condition on the parameters defining
$H$: a necessary condition for $S$ to exist is that $2\Re(a)\Im(a)+\Im(bc)=0$. When this condition is satisfied, the explicit expression of $S$
is fixed solving the above linear equation for $\Phi$.

\vspace{2mm}

It is interesting to observe that, not surprisingly, $H_{eff}$ in (\ref{triex}) satisfies condition $\det(X)=0$, and the matrix $S$ turns out
to be $\left(
                                                                                                                                           \begin{array}{cc}
                                                                                                                                             \sigma & s_r+is_i \\
                                                                                                                                              s_r-is_i & \sigma \\
                                                                                                                                           \end{array}
                                                                                                                                         \right),
$ where $\sigma$ must be chosen positive and such that $\sigma=\frac{1}{\delta}(\Im(\omega)s_r+\Re(\omega)s_i)$. Moreover,
$\sigma^2-(s_r^2+s_i^2)>0$ must also be satisfied.

Also, in \cite{jon}, the following hamiltonian is considered:
$$
H=\left(
    \begin{array}{cc}
      a & i\,b \\
      i\,b & -a \\
    \end{array}
  \right),
$$
with $a, b\in\Bbb{R}$. Again, it is easy to check that $\det(X)=0$ and that $S$ must have the following form $$S=\left(
    \begin{array}{cc}
      \sigma_{11} & \frac{b}{2a}(\sigma_{11}+\sigma_{22}) \\
      -\frac{b}{2a}(\sigma_{11}+\sigma_{22}) & \sigma_{22} \\
    \end{array}
  \right),$$
with $\sigma_{11}$ and $\sigma_{22}$ both positive, and satisfying the inequality
$\frac{\sigma_{11}}{\sigma_{22}}+\frac{\sigma_{22}}{\sigma_{11}}>\frac{2}{b^2}(2a^2-b^2)$.

\end{document}